\documentclass[11pt]{article}%
\usepackage{amsfonts}
\usepackage{amsmath,bbm,dsfont,mathrsfs,mathtools,appendix}
\usepackage{amssymb}
\usepackage{todonotes}
\usepackage{graphicx}
\usepackage{fullpage}%
\usepackage[colorlinks=true,linkcolor=blue,citecolor=red,plainpages=false,pdfpagelabels]%
{hyperref}%
\providecommand{\U}[1]{\protect\rule{.1in}{.1in}}

\newtheorem{theorem}{Theorem}

\newtheorem{lemma}[theorem]{Lemma}

\newtheorem{remark}[theorem]{Remark}

\newenvironment{proof}[1][Proof]{\noindent\textbf{#1.} }{\ \rule{0.5em}{0.5em}}

\newcommand{\cB}{{\mathcal{B}}}

\newcommand{\cH}{\mathcal{H}}

\newcommand{\N}{\mathbbm{N}}

\newcommand{\nn}{\nonumber}

\newcommand{\cP}{\mathcal{P}}

\newcommand{\cM}{\mathcal{M}}
\newcommand{\cN}{\mathcal{N}}

\newcommand{\1}{\mathbbm{1}}
\newcommand{\cD}{\mathcal{D}}

\def\>{{\rangle}}
\def\<{{\langle}}
\newcommand{\be}{\begin{equation}}
\newcommand{\ee}{\end{equation}}
\newcommand{\bea}{\begin{eqnarray}}
\newcommand{\eea}{\end{eqnarray}}

\newcommand{\eps}{\varepsilon}

\newcommand{\kb}[1]{|#1\rangle\!\langle#1|} 

\newcommand{\Tr}{\mathrm{Tr}}

\newcommand{\comment}[1]{}

\newcommand{\err}{\operatorname{err}}

\newcommand{\tr}{{\rm Tr}}

\usepackage{color}
\definecolor{colorthree}{rgb}{0.01,0.51,0.93}

\title{Interpolating between symmetric and asymmetric hypothesis testing}

\allowdisplaybreaks

\begin{document}
{\author{ Robert Salzmann\thanks{University of Cambridge, Department of Applied
Mathematics and Theoretical Physics, Wilberforce Road, Cambridge CB3 0WA,
United Kingdom}
\and Nilanjana Datta\footnotemark[1]}}
\maketitle
\begin{abstract}
The task of binary quantum hypothesis testing is to determine the state of a quantum system via measurements on it, given the side information that it is in one of two possible states, say $\rho$ and $\sigma$. This task is generally studied in either the {\em{symmetric setting}}, in which the two possible errors incurred in the task (the so-called type I and type II errors) are treated on an equal footing, or the {\em{asymmetric setting}} in which one minimizes the type II error probability under the constraint that the corresponding type I error probability is below a given threshold.
Here we define a one-parameter family of binary quantum hypothesis testing tasks, which we call $s$-{\em{hypothesis testing}}, and in which the relative significance of the two errors are weighted by a parameter $s$. In particular, $s$-hypothesis testing interpolates continuously between the regimes of symmetric and asymmetric hypothesis testing. Moreover, if arbitrarily many identical copies of the system are assumed to be available, then the minimal error probability of $s$-hypothesis testing is shown to decay exponentially in the number of copies, with a decay rate given by a quantum divergence which we denote as $\xi_s(\rho\|\sigma)$, and which satisfies a host of interesting properties. Moreover, this one-parameter family of divergences interpolates continuously between the corresponding decay rates for symmetric hypothesis testing (the quantum Chernoff divergence) for $s = 1$, and asymmetric hypothesis testing (the Umegaki relative entropy) for $s = 0$.  

\end{abstract}
\section{Introduction}
Discriminating between two states of a quantum system is a fundamental constituent of many quantum information theoretic tasks. Suppose a person (say, Bob) receives a quantum system $A$ which is in one of two possible states $\rho$ and $\sigma$. In order to infer what the actual state is, Bob does a measurement on $A$, which is most generally given by a POVM $\{\Lambda, \1-\Lambda\}$. In the language of binary hypothesis testing, one considers two hypotheses, the {\em{null hypothesis}} $H_0:\rho$ and the {\em{alternative hypothesis}} $H_1:\sigma$, and $0 \le \Lambda \le \1$ is referred to as a {\em{test}}.
Bob can make errors in his inference and the associated error probabilities are given as follows: $\alpha(\Lambda)\equiv {\mathbb{P}}(H_1|H_0) := \tr((\1-\Lambda)\rho)$ and $\beta(\Lambda)\equiv {\mathbb{P}}(H_0|H_1) := \tr(\Lambda\sigma)$, where ${\mathbb{P}}(H_1|H_0)$ denotes the probability of accepting $H_1$ when $H_0$ is true, and $\alpha(\Lambda)$ and $\beta(\Lambda)$ are referred to as the {\em{type I}} and {\em{type II}} error probabilities\footnote{They are also called the error
probabilities of the {\em{first}} and {\em{second kind}}, respectively.}, respectively. There is a tradeoff between these error probabilities and there are different ways of optimising them depending on the relative importance given to the two hypotheses.

In the setting of {\em{asymmetric hypothesis testing}}, one minimizes the type II error probability (over all possible POVMs) under the constraint that the type I error probability is not larger than  a given threshold value, say $\eps \in (0,1)$.
The minimal type II error probability under this constraint is hence given by
\begin{align}
\label{eq:SteinError}
\beta_\eps(\rho\|\sigma) \coloneqq \min\left\{\beta(\Lambda)\Big|\, \alpha(\Lambda) \le \eps,\,0\le\Lambda\le\1\right\}.
\end{align}

In contrast, in the setting of {\em{symmetric hypothesis testing}}, the two error probabilities are treated in a symmetric manner. In fact, one also takes into account the prior probabilities associated to the two hypothesis (i.e.~the probabilities that the states are $\rho$ and $\sigma$, respectively); if $p$ (resp.~$(1-p)$) is the prior probability of the hypothesis $H_0$ (resp.~$H_1$), then for any test $\Lambda$, one considers the minimal value of the  Bayesian error probability given by
   \begin{align}
p_{\err}(p,\rho,\sigma) \coloneqq \min_{0\le\Lambda\le\1} \Big(p\,\alpha(\Lambda) + (1-p)\beta(\Lambda)\Big).
\end{align}
One typical scenario where this setting is relevant is if one consider a quantum source which emits the state $\rho$ with probability $p$ and $\sigma$ with probability $1-p$. The goal is then to find the best possible measurement for discriminating the quantum states emitted by the source. On average, the minimal probability for making an error is then given by $p_{\err}(p,\rho,\sigma)$.

As in the classical case, quantum hypothesis testing was first studied in the so-called {\em{asymptotic i.i.d.~setting}} in which Bob receives multiple (say, $n$) identical copies of the system, instead of just one. He then does a joint measurement on all these $n$ systems in order to determine the state. The optimal asymptotic performance in the different settings (mentioned above) is quantified by the optimal exponential decay rates of the relevant error probabilities, evaluated in the asymptotic limit ($n \to \infty$). These are often called {\em{optimal error exponents}}, and in a series of seminal papers, they
have been shown to be given by two important {\em{quantum divergences.}}

In the asymmetric setting, Hiai and Petz~\cite{HiaiPetz_QuantumSteinLemma_1991} and Ogawa and Nagaoka~\cite{OgawaNagaoka_QuantumStein(StrongConverse)_2000} proved the quantum Stein's lemma which 
states that for all $\eps\in(0,1)$, 
$\beta_\eps(\rho^{\otimes n}\|\sigma^{\otimes n})$ decays exponentially to zero in the limit $n\to\infty$ with rate given by the {\em{Umegaki relative entropy (or divergence)}} \cite{Umegaki_ConditionalExpectation_1962}, i.e.
\begin{align}
\label{eq:QStein}
\lim_{n\to\infty}\frac{-\log\left(\beta_\eps(\rho^{\otimes n}\|\sigma^{\otimes n})\right)}{n} = D(\rho\|\sigma),
\end{align}
where $D(\rho\|\sigma)\coloneqq \tr(\rho(\log \rho - \log \sigma))$ if ${\rm{supp}} (\rho) \subseteq {\rm{supp}} (\sigma)$ and is set equal to infinity else. 

In the symmetric setting, Nussbaum and Szkola~\cite{Nussbaum_chernoff_2009}, and Audenaert {\em{et al.}}~\cite{Audenaert_QuantumChernoff_2007} proved that for all $p\in(0,1)$, $p_{\rm{err}}(p, \rho^{\otimes n}, \sigma^{\otimes n})$ decays exponentially to zero in the limit $n\to\infty$ with rate given by the so-called {\em{quantum Chernoff divergence}}

\begin{align}
\label{QCB} 
\lim_{n \to \infty} \frac{-\log \left(p_{\rm{err}}( \rho^{\otimes n}, \sigma^{\otimes n})\right)}{n}= \xi(\rho,\sigma),
\end{align}
where $\xi(\rho,\sigma) := \sup_{0\le\alpha\le 1} -\log\left(\Tr\left(\rho^{\alpha}\sigma^{1-\alpha}\right)\right).$
These seminal results are quantum analogues of their classical counterparts, for which the above divergences are replaced by the Kullback-Leibler divergence~\cite{CoverThomas_InformationTheory(Book)_2006} and the (classical) Chernoff divergence~\cite{Chernoff_AsymptoticEfficiencyforTests_1952}, respectively. 

\medskip

\noindent
A natural question to ask is the following: 
{\em{Is there a way to interpolate between symmetric and asymmetric hypothesis testing in an operational manner? }}
\smallskip

\noindent
In this paper, we answer the question affirmatively and prove that this interpolation scheme also leads to the definition of a one-parameter family of quantum divergences, which reduce to the quantum relative entropy and to the quantum Chernoff divergence for the two extreme values of the parameter. Our interpolation arises from the consideration of the type II error probability under a constraint on the type I error probability which depends on the parameter (say, $s$) as well as the type II error probability.

For the following discussion it will be useful to define a notion of equivalence between different hypothesis tesing errors\footnote{Here we use the term 'hypothesis testing error' in a broad sense to refer to any  real valued function $\gamma$ on pairs of states. This includes the optimal error probability for any binary quantum hypothesis testing task.}. We say two hypothesis testing errors $\gamma_{1}(\rho\|\sigma)$ and $\gamma_{2}(\rho\|\sigma)$ are equivalent, denoted by $\gamma_1 \sim \gamma_2$, if there exist constants $c,C>0$ such that for all quantum states $\rho,\sigma$
\begin{align}
c\,\gamma_2(\rho\|\sigma) \le \gamma_1(\rho\|\sigma) \le C \gamma_2(\rho\|\sigma).
\end{align}
Clearly, if $\gamma_1\sim\gamma_2$ and $\gamma_1(\rho^{\otimes n}\|\sigma^{\otimes n})$ decays exponentially to zero for $n\to\infty$ with some exponential rate, then also $\gamma_2(\rho^{\otimes n}\|\sigma^{\otimes n})$ decays exponentially to zero with the same rate. On the other hand, the fact that two hypothesis testing errors have the same exponential decay rate in the i.i.d setting does not imply that they are equivalent. To see this note that in the setting of asymmetric hypothesis testing for all $\eps_1,\eps_2\in(0,1)$ with $\eps_1\neq\eps_2$ the corresponding minimal error probabilities $\beta_{\eps_1}$ and $\beta_{\eps_2}$ are not equivalent (cf. Lemma~\ref{lem:QsCEquiv}). However as mentioned above, they have the same exponential decay rate, since the strong converse property in quantum Stein's lemma holds \eqref{eq:QStein}.

In the setting of symmetric hypothesis testing it is easy to see that for all $p,q\in(0,1)$ the corrresponding minimal error probabilities in symmetric hypothesis testing are equivalent, i.e. $p_{\err}(p,\cdot,\cdot) \sim p_{\err}(q,\cdot,\cdot).$ 
Moreover, in \cite[Lemma 4.10]{SalzDatta_RTSD_2021} it was shown that for all $p\in(0,1)$ the minimal error probability, $p_{\err}(p,\cdot,\cdot)$, is also equivalent to the hypothesis testing error probability $Q_{\min}$ defined as
\begin{equation}
\label{eq:Qmin}
Q_{\min}(\rho,\sigma) := \min\left\{\beta(\Lambda)\Big|\,  \alpha(\Lambda) \le \beta(\Lambda),\,0\le\Lambda\le\1\right\}.\footnote{Note that in \cite{SalzDatta_RTSD_2021} $Q_{\min}$ was defined  with an additional prefactor of $2$ because of the resource theoretic perspective taken in that paper. Here, however, we omit this prefactor to make the similarity to \eqref{eq:SteinError} more apparent. Moreover, note that as discussed in \cite{SalzDatta_RTSD_2021}, $Q_{\min}(\rho,\sigma)$ is symmetric in $\rho$ and $\sigma$ and that the inequality in the constraint in \eqref{eq:Qmin} can actually be replaced by an equality.}
\end{equation}
This implies, in particular, that the optimal error exponent in the symmetric setting (given by the left hand side of \eqref{QCB}) can be expressed in terms of $Q_{\min}$, and we have
\begin{align}
     \lim_{n \to \infty}  \frac{-\log \left(p_{\err}( \rho^{\otimes n}, \sigma^{\otimes n})\right)}{n} \equiv \lim_{n \to \infty}  \frac{-\log \left(Q_{\min}( \rho^{\otimes n}, \sigma^{\otimes n})\right)}{n}&= \xi(\rho,\sigma),
\end{align}

The similarity between the minimum error probability $\beta_\eps$ (given by~\eqref{eq:SteinError}) in the asymmetric setting, and the quantity $Q_{\min}$ (given by~\eqref{eq:Qmin}) in the symmetric setting, leads us naturally to define {\em{a continuous one-parameter family of hypothesis testing tasks}} : 
For $s\ge 0$\footnote{Note that for $s<0$, $Q^{(s)}(\rho\|\sigma)=0$ for all states $\rho$ and $\sigma$, and hence the range $s<0$ is not of interest.},
let us define the minimal type II error probability such that the type I error probability is upper bounded by the $s^{th}$ power of the type II error probability as follows: 
\begin{align}
\label{eq:DefQs}
 Q^{(s)}(\rho\|\sigma) &:=  \min\left\{\beta(\Lambda)\Big|\,  \alpha(\Lambda) \le \beta(\Lambda)^s,\,0\le\Lambda\le\1\right\}
\end{align}

More generally, we can define for any $C\ge 0$, the following minimal type II error probability:
\begin{align}
\label{eq:DefQCs}
 Q_C^{(s)}(\rho\|\sigma) &:= \min\left\{\beta(\Lambda)\Big|\,  \alpha(\Lambda) \le C\beta(\Lambda)^s,\,0\le\Lambda\le\1\right\},
 \end{align}
which reduces to  $Q^{(s)}(\rho\|\sigma)$ for $C=1$. In particular, for $\eps>0$ and $s\ge 0$ taking $C= \eps^{1-s}$, the corresponding $Q_{\eps^{1-s}}^{(s)}$ interpolates between the asymmetric hypothesis testing error probability $\beta_\eps$ at $s=0$ and the error probability, $Q_{\min}$, associated with symmetric hypothesis testing at $s=1$. 
However, in Lemma~\ref{lem:QsCequalities} we prove that for any $s >0$, the error probabilities $Q_C^{(s)}$ are equivalent for all $C>0$.

Let us try to get an intuitive understanding of the family of minimal error probabilities, $Q^{(s)}$, defined above. By comparing $Q^{(s)}$ to $\beta_\eps$ and $Q_{\min}$, it is clear that each value of $s$ corresponds to a different weighting of the relative significance of type I and type II errors, in the sense that smaller the value of $s$, the more one wants to avoid a type II error compared to a type I error. To further motivate the particular form of $Q^{(s)}$ let us consider a medical example\footnote{Note that quantum hypothesis testing generalises classical hypothesis testing and hence the definition of $Q^{(s)}$ can equivalently be applied to a classical binary hypothesis testing task.} for the particular value $s=1/2$:\footnote{By a similar argument one can also get an intuitive understanding of $Q^{(s)}$ for $s=1/n$ for all natural numbers $n$.} Imagine one wants to construct a test to check whether a patient has a certain disease or not. Here, the null hypothesis $H_0$ is that the patient is healthy whereas the alternative hypothesis $H_1$ is that the patient is ill. Of course, incorrectly concluding that the patient is healthy when the patient is actually ill (i.e.~a false negative result) is worse than the other way around (i.e.~a false positive result). Hence, one is interested in constructing a test $\Lambda$ which has a smaller type II error probability than type I error probability, i.e.~$\beta(\Lambda)\le\alpha(\Lambda)$. On the other hand, a positive test results in the patient undergoing a treatment which might have damaging side effects. Imagine now that at the end of the treatment the test is applied again, and the treatment is extended if the test result came out positive for a second time. We assume for simplicity that both tests (before and after the treatment) are the same, and that if the patient was actually healthy but the first test gave a false positive result, the probability that also the second test gives a false positive result is the same. Now, due to the adverse effects of the treatment, it might be worse to get two false positive results in a row, which occurs with probability $\alpha(\Lambda)^2$, than getting a false negative result in the first test, which occurs with probability $\beta(\Lambda)$. In this case, we would want to optimise over all tests $\Lambda$ satisfying the constraint $\alpha(\Lambda)^2\le\beta(\Lambda)\le \alpha(\Lambda).$ Note that $Q^{(1/2)}$ exactly gives the minimal type II error probability under this constraint.\footnote{Note that the constraint $\beta(\Lambda)\le\alpha(\Lambda)$ will be trivially fullfilled for the optimal test in the minimisation of $Q^{(1/2})$ given in \eqref{eq:DefQs}. Mathematically this can be seen by noting that the inequality constraint in \eqref{eq:DefQs} can actually be replaced by an equality as it is shown in Lemma~\ref{lem:QsCequalities}.}

The above definitions motivate us to coin the term $s$-{\em{hypothesis testing}}: for any given $s >0$, it is the task of binary quantum hypothesis testing with hypotheses 
$H_0:\rho$ and $H_1:\sigma$ such that the minimal error probability of interest is given by $Q^{(s)}(\rho\|\sigma)$ (or equivalently by  $Q_C^{(s)}(\rho\Vert \sigma)$). We evaluate the optimal error exponents of $s$-hypothesis testing in the asymptotic i.i.d.~setting, and show that they are given by a continuous family of quantum divergences parametrized by $s > 0$. We denote these by $\xi_s(\rho\Vert \sigma)$. These are defined by \eqref{eq:xi-s} of Section~\ref{sec:xi-s} and are shown to satisfy a host of interesting properties (see Lemma~\ref{lem:xi-s-props}). Most interestingly, they converge to the Umegaki relative entropy and to the quantum Chernoff divergence in the limits $s \to 0$ and $s \to 1$, respectively. Hence, the task of $s$-hypothesis testing interpolates between asymmetric and symmetric hypothesis testing both in the one-shot\footnote{That is, when Bob is given a single copy of the quantum system $A$, which is either in the state $\rho$ or $\sigma$.} and in the asymptotic i.i.d.~setting.
\begin{center}
\begin{figure}[t]
  \includegraphics[width=0.85\textwidth]{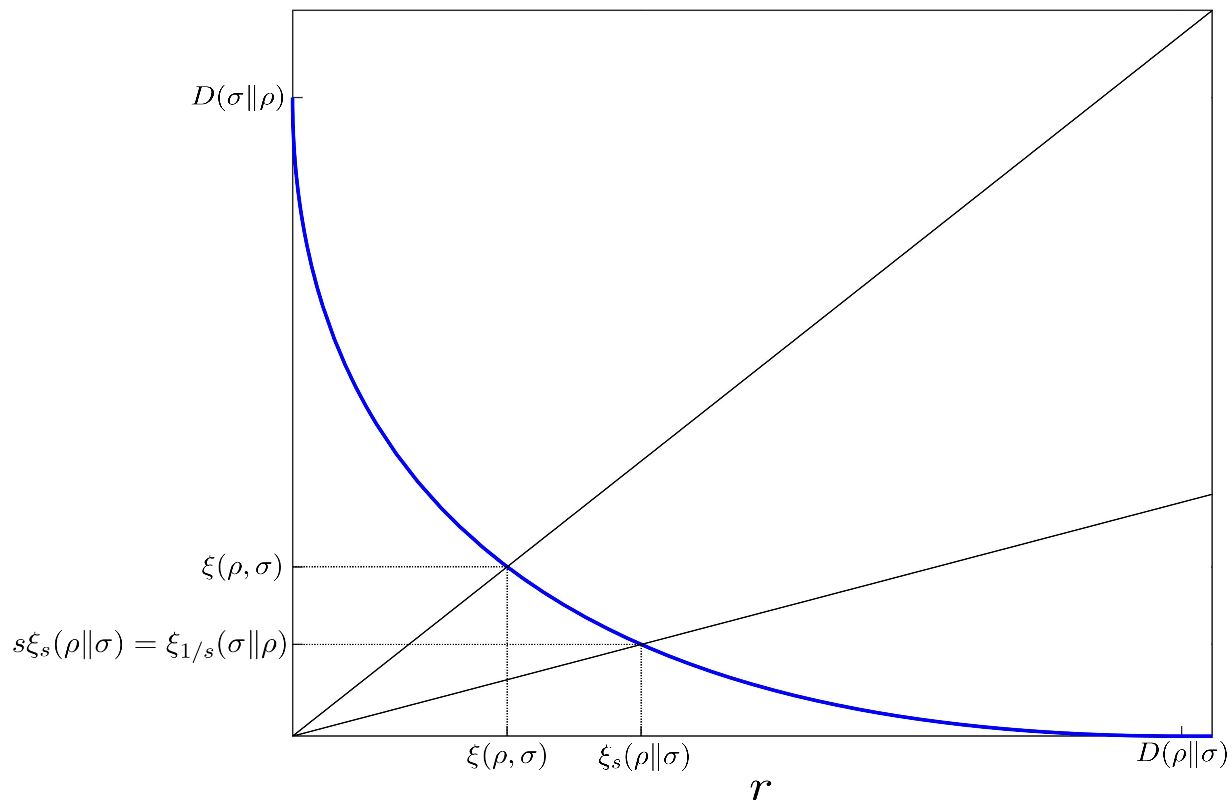}
	\caption{Here we plot in blue the right-hand side of \eqref{eq:HoeffResult} which equals $B(r|\rho\|\sigma)$ for $r>0$. We consider randomly generated states $\rho$ and $\sigma$ which have full support and hence $D(\rho\|\sigma),D(\sigma\|\rho)<\infty$. In black we plot linear graphs with slope $1$ and $s=1/3$, respectively. The points of intersection of these lines with the blue curve have $r$-coordinates equal to 
	$\xi(\rho,\sigma)$ and $\xi_s(\rho\|\sigma)$, respectively (cf. Theorem~\ref{thm:AEPQs} for definition of $\xi_s$ and points~\ref{one}-\ref{three} below as well as Lemma~\ref{lem:xi-s-props} and~\ref{lem:xisHoef} for the illustrated relations between the different divergences).}
	\label{fig:Hoeffding}
\end{figure}
\end{center}

\begin{figure}[t]
\centering
  \includegraphics[width=0.75\textwidth]{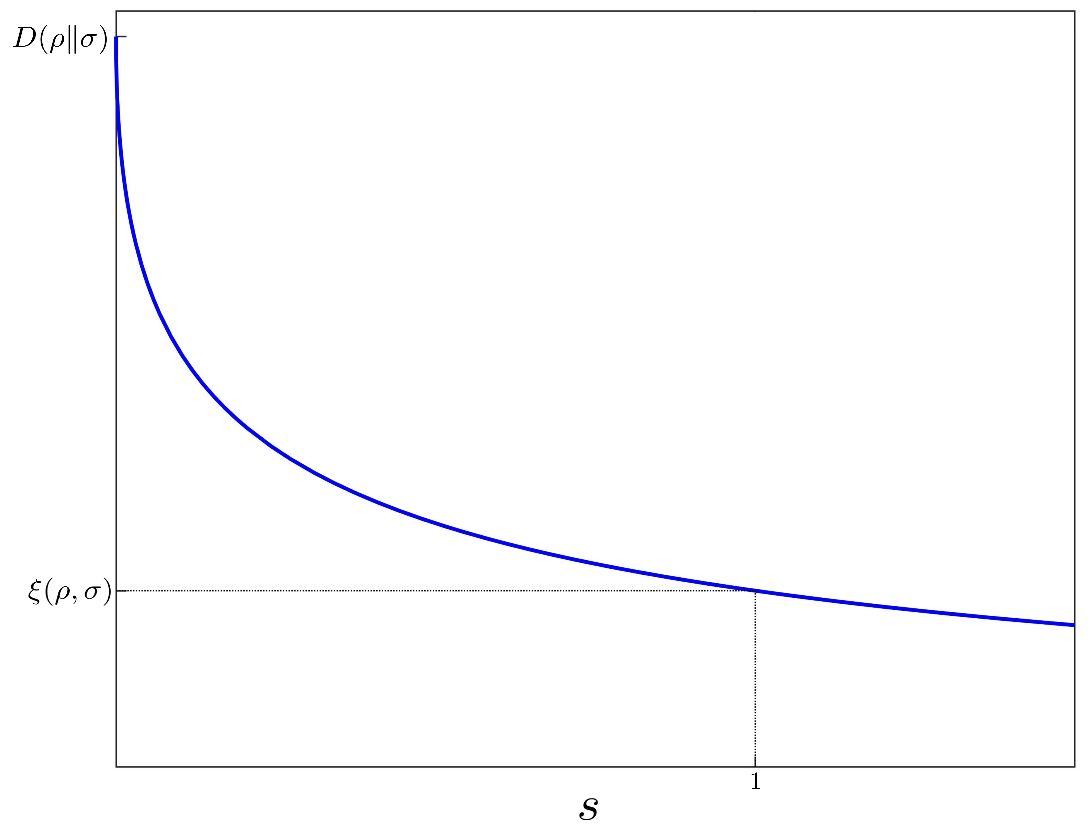}
	\caption{Here we plot $\xi_s(\rho\|\sigma)$ (see Theorem~\ref{thm:AEPQs} for definition) with $\rho$ and $\sigma$ being the same states as in Figure~\ref{fig:Hoeffding}. We denote $D(\rho\|\sigma)$ and $\xi(\rho,\sigma)$ at $s=0$ and $s=1$ respectively (cf. Lemma~\ref{lem:xi-s-props}).}
\end{figure}

Before proceeding further, let us recall another important hypothesis testing scenario which has been studied exhaustively in both classical and quantum binary hypothesis testing, and which characterizes the set of possible type I and type II error exponents. 
In this scenario, one evaluates the optimal type I error exponent (i.e.~the optimal exponential decay rate of the type I error probability), evaluated in the asymptotic limit, under the constraint that the type II error exponent is not smaller than a given threshold value (say, $r$):
\begin{align}
\label{eq:HoeffdingDef}
B(r|\rho\|\sigma) &\coloneqq \sup_{\substack{(\Lambda_n)_{n\in\N}\\0\le\Lambda_n\le\1}}\left\{\liminf_{n\to\infty}\frac{-\log\Big(\Tr\left((\1-\Lambda_n)\rho^{\otimes n}\right)\Big)}{n}\Bigg|\liminf_{n\to\infty}\frac{-\log\Big(\Tr\left(\Lambda_n\sigma^{\otimes n}\right)\Big)}{n} \ge r,\right\}
\end{align}
The above quantity is known as the {\em{quantum Hoeffding bound}} and it was shown by Hayashi \cite{Hayashi_Hoeffding(ErrorExponent)_2007} and Nagaoka \cite{Nagaoka_ConverseHoeffding_2006} (also consider \cite{AudenaertNussbaum_AsymptoticErrorRates_2008})  that for $r>0$ the optimal error exponent defined above is given by the following expression
\begin{align}
\label{eq:HoeffResult}
    B(r|\rho\|\sigma)&= \sup_{0\le\alpha\le 1} \frac{\alpha-1}{\alpha}\left(r - D_\alpha(\rho\|\sigma)\right),
    \end{align}
    where 
    $D_\alpha(\rho\|\sigma):= \frac{1}{\alpha - 1} \log (\tr(\rho^\alpha \sigma^{1-\alpha}))$
is the Petz-R\'enyi relative entropy of order $\alpha$ \cite{Petz_Quasi-entropies_1986}. From that it is easy to see that $B(r|\rho\|\sigma)<\infty$ if and only if $r>D_{\min}(\rho\|\sigma)$ where the min-relative entropy \cite{Datta_MinMaxRelativeEntr_2009} is defined by $D_{\min}(\rho\|\sigma)= -\log(\Tr(\pi_\rho \sigma))$ with $\pi_\rho$ being the projector onto the support of $\rho$ (cf. Lemma~\ref{lem:AppFinHoeff}). The classical analogue of~\eqref{eq:HoeffdingDef} was studied, as the name suggests, by Hoeffding~\cite{Hoeffding_AsympOptimalTests_1965}. It can be shown using~\eqref{eq:HoeffResult} that (see also~\cite{Hayashi_Hoeffding(ErrorExponent)_2007} and \cite{Nagaoka_ConverseHoeffding_2006}):
\begin{enumerate}
    \item\label{one} for $\rho\neq\sigma$\footnote{For $\rho=\sigma$ and hence $\xi(\rho,\sigma) =0$ the point \ref{one} is no longer true as $B(0|\rho\|\sigma) =\infty$. Then $\xi(\rho,\sigma)$ is just the infimum over all $r$ solving $B(r|\rho\|\sigma) 
    =r$.} the unique solution of the equation $B(r|\rho\|\sigma) =r $ is equal to the quantum Chernoff divergence $\xi(\rho, \sigma)$;
    \item\label{two} the infimum over all $r$ solving the equation $B(r|\rho\|\sigma) =0 $ is equal to the Umegaki relative entropy $D(\rho\Vert \sigma)$;
    \item \label{three} the limit $r\to 0$ recovers the Umegaki relative entropy with $\rho$ and $\sigma$ in reversed order, i.e. $\lim_{r\to 0} B(r|\rho\|\sigma) = D(\sigma\|\rho)$.
\end{enumerate}

Hence, the {\em{optimal error exponents}} of the asymmetric and symmetric setting can be recovered from the quantum Hoeffding bound. See Figure~\ref{fig:Hoeffding} for an illustration. Therefore, in a sense, the quantum Hoeffding bound also provides a way to interpolate between symmetric and asymmetric hypothesis testing. However, unlike the task of $s$-hypothesis testing introduced above, this interpolation is only at the level of optimal error exponents and is hence restricted  to the asymptotic i.i.d.~setting.

\subsection*{Main Results and layout of the paper}

Suppose Bob receives a finite-dimensional quantum system, $A$, with associated Hilbert space $\cH$ which is in one of two states $\rho$ and $\sigma$. Let $\cB(\cH)$ denote the set of linear operators on $\cH$, and $\cP(\cH)$ denote the set of positive semi-definite operators 
on $\cH$. Then $\rho, \sigma \in \cD(\cH)$, where the latter denotes the set of density matrices, i.e.~the set of positive semi-definite operators with unit trace. \comment{Bob's task is to determine
which of the two states it is in. In the language of binary quantum hypothesis testing the null and alternative hypotheses are $H_0:\rho$ and $H_1:\sigma$, respectively. In the task of $s$-{\em{hypothesis testing}} for any $s\ge 0$ the aim is to minimize the type II error probability $\beta(\Lambda) \equiv \tr(\Lambda \sigma)$, where $0\le \Lambda \le I$ denotes a test, under the constraint that the corresponding type I error probability  
$\alpha(\Lambda) \equiv \tr((I-\Lambda) \rho)$ is such that $\alpha(\Lambda)  \leq  \beta(\Lambda)^s$. Hence, in the one-shot setting, i.e.~when Bob receives a single copy of the system $A$, the quantity of interest is given by $Q^{(s)}(\rho\|\sigma)$, defined in~\eqref{eq:DefQCs}. In contrast, in the asymptotic i.i.d.~setting, in which Bob receives arbitrarily many (say, $n$) identical copies of the system $A$, the optimal performance of the $s$-hypothesis testing task
is quantified by the optimal exponential decay rate of $Q^{(s)}(\rho^{\otimes n}\|\sigma^{\otimes n})$. This is called the {\em{optimal error exponent}} of $s$-hypothesis testing and is given by
$$\liminf_{n\to\infty}\frac{-\log\left(Q^{(s)}(\rho^{\otimes n}\|\sigma^{\otimes n})\right)}{n}. $$}

\begin{itemize}
    \item Our main result, given by the following theorem, concerns the evaluation of the {\em{optimal error exponent}} of $s$-hypothesis testing which is given by
\begin{align}
\label{eq:OptExponent}
\liminf_{n\to\infty}\frac{-\log\left(Q^{(s)}(\rho^{\otimes n}\|\sigma^{\otimes n})\right)}{n}.
\end{align}
We show that the limit in the above expression exists and is given by a quantum divergence which we denote as $\xi_s(\rho\Vert\sigma)$.
\begin{theorem}
\label{thm:AEPQs}
For all $s> 0$ and states $\rho,\sigma$ we have
	\begin{align}
	\lim_{n\to\infty}\frac{-\log\left( Q^{(s)}(\rho^{\otimes n}\|\sigma^{\otimes n})\right)}{n} = \xi_s(\rho\|\sigma),
	\end{align}
	where
	$$
	   \xi_s(\rho\|\sigma) := \sup_{0\le \alpha\le 1}\frac{\log\left(\Tr\left(\rho^\alpha\sigma^{1-\alpha}\right)\right)}{\alpha(1-s)-1}.
	$$
\end{theorem}
This is proved in Section~\ref{sec:s-HT}.
\item In Section~\ref{sec:xi-s} we prove a host of interesting properties of the quantity $\xi_s(\rho\|\sigma)$, including the following: $(i)$ it is
indeed a quantum divergence, in the sense that it is non-negative and satisfies the so-called data-processing inequality, i.e.~monotonicity under the 
action of quantum channels, $(ii)$ it interpolates between the Umegaki relative entropy $D(\rho\Vert \sigma)$ and the quantum Chernoff divergence $\xi(\rho, \sigma)$ (see Lemma~\ref{lem:xi-s-props}), and $(iii)$ its relation with the quantum Hoeffding bound (see Lemma~\ref{lem:xisHoef}).
\end{itemize}

\section{A one-parameter family of quantum divergences $\xi_s(\rho\Vert \sigma)$}
\label{sec:xi-s}

For a pair of quantum states $\rho$ and $\sigma$ of a finite-dimensional quantum system, and $s\ge 0$\footnote{Note that the restriction $s\ge0$ is meaningful as for $s<0$ the corresponding $\xi_s(\rho\|\sigma)$ is infinite.} we define the following one-parameter family of quantum divergences:
\begin{align}
\label{eq:xi-s}
\xi_s(\rho\|\sigma) := \sup_{0\le \alpha\le 1}\frac{\log\left(\Tr\left(\rho^\alpha\sigma^{1-\alpha}\right)\right)}{\alpha(1-s)-1}.
\end{align}	
These divergences are of important operational significance in the information-theoretic task of $s$-hypothesis testing which was mentioned in the Introduction and is elaborated in Section~\ref{sec:s-HT}. In particular, the optimal error exponent of $s$-hypothesis testing, defined in~\eqref{eq:OptExponent}, is shown to be given by $\xi_s(\rho\|\sigma)$ (see Theorem~\ref{thm:AEPQs}). In the following lemma we list some interesting properties of $\xi_s(\rho\|\sigma)$. These include  the data-processing inequality~\eqref{eq:DPIxis} (which justifies it being called a quantum divergence), and the facts that in the limits $s \to 0$ and $s \to 1$ it reduces to the Umegaki relative entropy and the quantum Chernoff divergence, respectively.

\begin{lemma} [General properties of $\xi_s$]
\label{lem:xi-s-props}$ $
\begin{enumerate}
\item \label{functionalproperties} The map $s\mapsto\xi_s(\rho\|\sigma)$ is monotonically decreasing, convex and continuous.
In particular, continuity gives 
\begin{align}
\label{eq:LimitRel}
\lim_{s\to 0} \xi_s(\rho\|\sigma) &= \xi_0(\rho\|\sigma)= D(\rho\|\sigma),\\
\lim_{s\to 1}\xi_s(\rho\|\sigma) &=\xi_1(\rho\|\sigma)= \xi(\rho,\sigma),\label{eq:LimitChern}
\end{align}
where $D(\cdot\|\cdot)$ denotes the Umegaki relative entropy and $\xi(\cdot,\cdot)$ the quantum Chernoff divergence.
 Moreover, for any $c>0$ and $s_1,s_2\ge c$ and states $\rho$ and $\sigma$ having mutually non-orthogonal supports we have the Lipschitz continuity bound
\begin{align}
\label{eq:Lipsch}
|\xi_{s_1}(\rho\|\sigma) - \xi_{s_2}(\rho\|\sigma)| \le C|s_1-s_2|,
\end{align}
where $C\ge0$ only depends on $\rho,\sigma$ and $c$ but not on $s_1,s_2$.
\item We have the reciprocity relation
\begin{align}
\label{eq:xissymm}
\xi_s(\rho\|\sigma) =\frac{1}{s} \xi_{1/s}(\sigma\|\rho).
\end{align}
\item \label{xi-s-DivProp} For all $s\ge0$ and states $\rho,\sigma$ we have $\xi_s(\rho\|\sigma)\ge 0$, and $\xi_s(\rho\|\sigma) = 0$ if and only if $\rho=\sigma.$ Moreover, the map $(\rho,\sigma) \mapsto \xi_s(\rho\|\sigma)$ on the set $\cP(\cH)\times\cP(\cH)$ is jointly convex. In particular this gives that $\xi_s$ satisfies the data-processing inequality, i.e.~for all states $\rho,\sigma$ and linear completely positive trace-preserving (CPTP) maps 
$\cN$ on $\cP(\cH)$ we have
\begin{align}
\label{eq:DPIxis}
\xi_s(\cN(\rho)\|\cN(\sigma)) \le \xi_s(\rho\|\sigma). 
\end{align}
\end{enumerate}
\end{lemma}
\begin{remark}
Note that \eqref{eq:LimitRel} together with \eqref{eq:xissymm} gives
\begin{equation*}
\lim_{s\to\infty}s\,\xi_s(\rho\|\sigma) = D(\sigma\|\rho),
\end{equation*}
which is the analogous relation to point \ref{three} for the $r\to 0$ limit of $B(r|\rho\|\sigma)$.
\end{remark}

\begin{proof}[Proof of Lemma~\ref{lem:xi-s-props}]
    We start with the proof of \ref{functionalproperties}. Since $\log\left(\Tr(\rho^{\alpha}\sigma^{1-\alpha})\right)\le 0$ for all $\alpha\in[0,1]$, we see that the maps $s \mapsto \frac{\log\left(\Tr\left(\rho^\alpha\sigma^{1-\alpha}\right)\right)}{\alpha(1-s)-1}$ and $s\mapsto \xi_s(\rho\|\sigma)$ are monotonically decreasing. Convexity of $\xi_s(\rho\|\sigma)$ in $s$ follows from convexity of $s\mapsto (1-\alpha(1-s))^{-1}.$
	For the Lipschitz bound \eqref{eq:Lipsch} let $c>0$ and  $s_1,s_2\ge c$. Then \eqref{eq:Lipsch} follows by noting that
	\begin{align*}
	|\xi_{s_1}(\rho\|\sigma) - \xi_{s_2}(\rho\|\sigma)| &\nn= \left|\sup_{0\le \alpha\le 1}\frac{\log\left(\Tr\left(\rho^\alpha\sigma^{1-\alpha}\right)\right)}{\alpha(1-s_1)-1} - \sup_{0\le \alpha\le 1}\frac{\log\left(\Tr\left(\rho^\alpha\sigma^{1-\alpha}\right)\right)}{\alpha(1-s_2)-1}\right| \\&\le\sup_{0\le\alpha\le 1} \left|\frac{\log\left(\Tr\left(\rho^\alpha\sigma^{1-\alpha}\right)\right)}{\alpha(1-s_1)-1} - \frac{\log\left(\Tr\left(\rho^\alpha\sigma^{1-\alpha}\right)\right)}{\alpha(1-s_2)-1}\right| \nn\\&=\sup_{0\le\alpha\le 1} \left|\frac{\alpha\log\left(\Tr\left(\rho^\alpha\sigma^{1-\alpha}\right)\right)}{(\alpha(1-s_2)-1)(\alpha(1-s_1)-1)}\right| |s_1 -s_2|\nn
	\\&\le \sup_{0\le\alpha\le1}\left|\frac{\alpha\log\left(\Tr\left(\rho^\alpha\sigma^{1-\alpha}\right)\right)}{(\alpha(1-c)-1)(\alpha(1-c)-1)}\right| |s_1 -s_2|.
	\end{align*}
	and using the fact that $C\coloneqq\sup_{0\le\alpha\le1}\left|\frac{\alpha\log\left(\Tr\left(\rho^\alpha\sigma^{1-\alpha}\right)\right)}{(\alpha(1-c)-1)(\alpha(1-c)-1)}\right|$ is finite as long as $\rho$ and $\sigma$ have mutually non-orthogonal supports.
	This shows continuity of $s\mapsto\xi_s(\rho\|\sigma)$ on $(0,\infty)$, and then~\eqref{eq:LimitChern} follows since
	\begin{align*}
	\xi_1(\rho\|\sigma) = \sup_{0\le\alpha\le 1}\Big(-\log\left(\Tr\left(\rho^\alpha\sigma^{1-\alpha}\right)\right)\Big) = \xi(\rho,\sigma).
	\end{align*}
	To prove~\eqref{eq:LimitRel}, note that
	\begin{align*}
	\lim_{s\to 0}\xi_s(\rho\|\sigma) &= \sup_{s\ge 0} \xi_s(\rho\|\sigma) = \sup_{s\ge 0}\sup_{0\le \alpha\le 1}\frac{\log\left(\Tr(\rho^\alpha\sigma^{1-\alpha}\right)}{\alpha(1-s)-1} =\sup_{0\le \alpha\le 1}\sup_{s\ge 0} \frac{\log\left(\Tr(\rho^\alpha\sigma^{1-\alpha}\right)}{\alpha(1-s)-1}\\&= \sup_{0\le \alpha\le 1} \frac{\log\left(\Tr(\rho^\alpha\sigma^{1-\alpha}\right)}{\alpha-1} = \sup_{0\le \alpha \le 1} D_{\alpha}(\rho\|\sigma) = D(\rho\|\sigma).
	\end{align*} 
		Here, we have used that the Petz-R\'enyi relative entropy $D_\alpha$ is monotonically increasing in $\alpha$ and satisfies $\lim_{\alpha\to 1}D_\alpha(\rho\|\sigma) = D(\rho\|\sigma)$ (see e.g. \cite{Tomamichel_FiniteResource(Book)_2016}).
		\smallskip

The relation \eqref{eq:xissymm} follows by noting that
\begin{align*}
\xi_s(\rho\|\sigma) = \sup_{0\le\alpha\le 1} \frac{\log\left(\Tr(\rho^{1-\alpha}\sigma^{\alpha})\right)}{(1-\alpha)(1-s) - 1} = \frac{1}{s} \sup_{0\le\alpha\le 1} \frac{\log\left(\Tr(\rho^{1-\alpha}\sigma^{\alpha})\right)}{\alpha(1-1/s)-1} = \frac{1}{s} \xi_{1/s}(\sigma\|\rho).
\end{align*}
\smallskip

Next we prove~\ref{xi-s-DivProp}. The first statement follows by noting that for all $\alpha\in[0,1]$ we have $\Tr(\rho^{\alpha}\sigma^{1-\alpha})\le 1$ by Hölder's inequality, with equality if and only if $\rho=\sigma.$ Since $(\rho,\sigma) \mapsto \Tr(\rho^{\alpha}\sigma^{1-\alpha})$ is jointly concave by Lieb's concavity theorem \cite{Lieb_ConvTraceFunctions(ConcTheo)_1973}, joint convexity of $(\rho,\sigma) \mapsto \xi_s(\rho\|\sigma)$ immediately follows. The data-processing inequality for $\xi_s$ then follows from joint convexity by the standard argument (see e.g.~\cite{FrankLieb_MonotonicityRenyiEntr_2013}). Alternatively the data-processing inequality can also be concluded by the data-processing inequality for the Petz-R\'enyi relative entropy $D_\alpha(\rho\Vert \sigma)$ \cite{Petz_Quasi-entropies_1986} and the identity
$\xi_s(\rho\|\sigma) = \sup_{0\le\alpha\le 1}\frac{(1-\alpha)D_\alpha(\rho\|\sigma)}{1-\alpha(1-s)}$.
\end{proof}

\begin{lemma}
	\label{lem:xisHoef}
	For all $s>0$ and $\rho\neq\sigma$ we have
	\begin{align}
	\label{eq:xisHoef}
	B(\xi_s(\rho\|\sigma) | \rho\|\sigma) = s\xi_s(\rho\|\sigma).
	\end{align}
	Moreover, in the above case, $\xi_s(\rho\|\sigma)$ is the unique solution to the equation $B(r | \rho\|\sigma) = sr$.
	In general, i.e. for all $s\ge 0$ and states $\rho,\sigma$, $\xi_s(\rho\|\sigma)$ is the infimum over all solutions to $B(r | \rho\|\sigma) = sr$.
\end{lemma}

\begin{proof}
Since for $\rho\neq \sigma$ we have $\xi_s(\rho\|\sigma)>0$, the relation \eqref{eq:HoeffResult} can be used. Moroever, assume without loss of generality that $\rho$ and $\sigma$ have mutually non-orthogonal supports, since otherwise $\xi_s(\rho\|\sigma)= \infty = B(\infty|\rho\|\sigma)$ is trivially satisfied. We have
	\begin{align}
	\nn B(\xi_s(\rho\|\sigma) | \rho\|\sigma) - s\xi_s(\rho\|\sigma) &= \sup_{0\le\alpha\le 1} \frac{\alpha-1}{\alpha}\Big(\xi_s(\rho\|\sigma) - D_{\alpha}(\rho\|\sigma)\Big)-s\xi_s(\rho\|\sigma) \nonumber\\
	&=\nn\sup_{0\le\alpha\le 1} \frac{1}{\alpha}\Big((\alpha(1-s)-1)\xi_s(\rho\|\sigma) - (\alpha-1)D_{\alpha}(\rho\|\sigma)\Big) \nonumber\\
	&=\nn\sup_{0\le\alpha\le 1} \frac{1}{\alpha}\Big((\alpha(1-s)-1)\sup_{0\le\beta\le 1}\frac{\log\left(\Tr\left(\rho^\beta\sigma^{1-\beta}\right)\right)}{\beta(1-s)-1}  - \log\left(\Tr\left(\rho^\alpha\sigma^{1-\alpha}\right)\right)\Big) \nonumber\\ 
	&=\nn\sup_{0\le\alpha\le 1}\inf_{0\le\beta\le 1} \frac{1}{\alpha}\Big((\alpha(1-s)-1)\frac{\log\left(\Tr\left(\rho^\beta\sigma^{1-\beta}\right)\right)}{\beta(1-s)-1}  - \log\left(\Tr\left(\rho^\alpha\sigma^{1-\alpha}\right)\right)\Big) \nonumber\\
	&\le\sup_{0\le\alpha\le 1} \frac{1}{\alpha}\Big((\alpha(1-s)-1)\frac{\log\left(\Tr\left(\rho^\alpha\sigma^{1-\alpha}\right)\right)}{\alpha(1-s)-1}  - \log\left(\Tr\left(\rho^\alpha\sigma^{1-\alpha}\right)\right)\Big)  = 0.
	\label{21}
	\end{align}	
	Here, we have used that $\alpha(1-s)-1\le 0$ to justify the equality in the fourth line.
	
	Since $s>0$, the map $\alpha \mapsto \frac{\log\left(\Tr(\rho^{\alpha}\sigma^{1-\alpha})\right)}{\alpha(1-s)-1}$ is continuous, and hence the supremum in \eqref{eq:xi-s} is actually a maximum, i.e.~there exists an $\alpha_0\in[0,1]$ attaining the finite supremum.
	Then, following on from the third line of~\eqref{21}, we get
	\begin{align*}
	\nn B(\xi_s(\rho\|\sigma) | \rho\|\sigma) - s\xi_s(\rho\|\sigma) &= \nn\nn\sup_{0\le\alpha\le 1} \frac{1}{\alpha}\Big((\alpha(1-s)-1)\frac{\log\left(\Tr\left(\rho^{\alpha_0}\sigma^{1-\alpha_0}\right)\right)}{\alpha_0(1-s)-1}  - \log\left(\Tr\left(\rho^\alpha\sigma^{1-\alpha}\right)\right)\Big)   \\ &\ge  \frac{1}{\alpha_0}\Big((\alpha_0(1-s)-1)\frac{\log\left(\Tr\left(\rho^{\alpha_0}\sigma^{1-\alpha_0}\right)\right)}{\alpha_0(1-s)-1}  - \log\left(\Tr\left(\rho^{\alpha_0}\sigma^{1-\alpha_0}\right)\right)\Big)  = 0.
	\end{align*}	
	Uniqueness follows since $r\mapsto B(r | \rho\|\sigma)$ is monotonically decreasing. 
	
	For $s=0$ the result follows by point~\ref{one} above: for $\rho=\sigma$, note that $\xi_s(\rho\|\sigma) = 0$ and $B(0|\rho\|\sigma)=\infty$ but $B(r|\rho\|\sigma)=0$ for all $r>0$. This implies that $\xi_s(\rho\|\sigma) =\inf\{r|\, B(r|\rho\|\sigma) = sr\}.$ 
\end{proof}
\section{A one-parameter family of hypothesis testing tasks: $s$-hypothesis testing}
\label{sec:s-HT}
In this section we study in detail the one-parameter family of binary quantum hypothesis testing tasks introduced in the Introduction, which we refer to as $s$-hypothesis testing,
with $s \geq 0$ being the parameter. We consider the task in somewhat more generality here: for the hypotheses $H_0:\rho$ and $H_1:\sigma$, we define the aim of $s$-hypothesis testing to be to minimize the type II error probability $\beta(\Lambda)\equiv \Tr(\Lambda\sigma)$, over all $0\le \Lambda \le I$, under the constraint that the corresponding type I error probability $\alpha(\Lambda)\equiv \Tr((\1-\Lambda)\rho)$ is at most equal to $C \beta(\Lambda)^s$ for some fixed $C>0$. 
Thus the minimal error probability of interest is 
$$ Q_C^{(s)}(\rho\|\sigma) := \min\left\{\beta(\Lambda)\Big|\,  \alpha(\Lambda) \le C\beta(\Lambda)^s,\,0\le\Lambda\le\1\right\}.$$ 
In the Introduction, we focussed on the particular case $C=1$, i.e.~on the minimal error probability $Q_1^{(s)}(\rho\|\sigma) \equiv Q^{(s)}(\rho\|\sigma)$. However, Lemma~\ref{lem:QsCEquiv} below establishes the equivalence $ Q^{(s)}\sim  Q_C^{(s)}$, and hence the optimal error exponents corresponding to $ Q^{(s)}$ and $Q_C^{(s)}$ are equal. Our ultimate aim is to prove our main result given in Theorem~\ref{thm:AEPQs} in the Introduction, which states that the optimal error exponent for $s$-hypothesis testing is equal to the quantum divergence $\xi_s(\rho\Vert \sigma)$ defined in~\eqref{eq:xi-s}. 
However, first we discuss some general properties of the minimal error probabilities $Q^{(s)}_C$ which we employ in the proof of this result.
\medskip

The following lemma shows that the inequality in the constraint in the definition of $ Q_C^{(s)}(\rho\|\sigma)$ can actually be taken to be an equality.
\begin{lemma}
	\label{lem:QsCequalities} For all $s,C>0$ and states $\rho,\sigma$ we have
	\begin{align}
	Q_C^{(s)}(\rho\|\sigma) =  \min\left\{\Tr(\Lambda\sigma)\Big|\,\, \Tr\left((\1-\Lambda)\rho\right)=C\left(\Tr(\Lambda\sigma)\right)^s,\,\,0\le\Lambda\le\1 \right\} = \left(Q_{1/C^{1/s}}^{(1/s)}(\sigma\|\rho)/C\right)^{1/s}.
	\label{eq:QSCequal}
	\end{align}
	In particular, for $C=1$ this gives the relation
	\begin{align}
	 Q^{(s)}(\rho\|\sigma) =\left(Q^{(1/s)}(\sigma\|\rho)\right)^{1/s}.
	\end{align}
\end{lemma}
\begin{proof}
	We begin by showing the first equality in \eqref{eq:QSCequal}, i.e.
	\begin{align}
	\label{eq:Qscfirstequal}
	Q_C^{(s)}(\rho\|\sigma) =  \min\left\{\Tr(\Lambda\sigma)\Big|\,\,  \Tr\left((\1-\Lambda)\rho\right)=C\left(\Tr(\Lambda\sigma)\right)^s,\,\,0\le\Lambda\le\1 \right\}.
	\end{align}
	Here, it only remains to show that $Q_C^{(s)}(\rho\|\sigma)$ is lower bounded by the right-hand side as the upper bound is immediate. Let $\Lambda_\star$ be a minimiser in \eqref{eq:DefQCs}, in particular
	\begin{align*}
	\Tr\left((\1-\Lambda_\star)\rho\right)\le C\left(\Tr(\Lambda_\star\sigma)\right)^s.
	\end{align*}
	By continuity there exists $t\in[0,1]$ such that $C\left(\Tr(t\Lambda_\star\sigma)\right)^s=\Tr\left((\1-t\Lambda_\star)\rho\right)$. Since
	\begin{align*}
	\Tr(t\Lambda_\star\sigma) \le \Tr(\Lambda_\star\sigma) = Q_C^{(s)}(\rho\|\sigma),
	\end{align*}
	it is clear that $Q_C^{(s)}(\rho\|\sigma)$ can be expressed by~\eqref{eq:Qscfirstequal}, i.e.~with the constraint being given by an equality.
	For the second equality in \eqref{eq:QSCequal} note that the established equation \eqref{eq:Qscfirstequal} applied to $Q_{1/C}^{(1/s)}(\sigma\|\rho)$ gives
	\begin{align*}
	\nn\left(Q_{1/C^{1/s}}^{(1/s)}(\sigma\|\rho)\right)^{1/s} &= \left(\min\left\{\Tr(\Lambda\rho)\Big|\,\, \left(\frac{\Tr\left(\Lambda\rho\right)}{C}\right)^{1/s}=\Tr((\1-\Lambda)\sigma),\,\,0\le\Lambda\le\1 \right\}\right)^{1/s} \\\nn &=\left(\min\left\{\Tr((\1-\Lambda)\rho)\Big|\,\, \Tr\left((\1-\Lambda)\rho\right)=C\left(\Tr(\Lambda\sigma)\right)^s,\,\,0\le\Lambda\le\1 \right\}\right)^{1/s} \\ &=\nn
	C^{1/s}\min\left\{\Tr(\Lambda\sigma)\Big|\,\,  \Tr\left((\1-\Lambda)\rho\right)=C\left(\Tr(\Lambda\sigma)\right)^s,\,\,0\le\Lambda\le\1 \right\} \\ &
	= C^{1/s} Q_C^{(s)}(\rho\|\sigma).
	\end{align*}
\end{proof}

In order to derive the optimal error exponent for  $s$-hypothesis testing, 
it will be useful to consider the following unconstrained optimisation problem:
\begin{align}\label{25}
p^{(s)}_{\err}(\rho,\sigma) := \min_{0\le\Lambda\le\1} \Big(\left(\Tr\left((\1-\Lambda)\rho\right)\right)^{1/s} + \Tr\left(\Lambda\sigma\right)\Big).
\end{align}
It is easy to show that $p^{(s)}_{\err}$ is monotonic under CPTP maps $\cN$, i.e.
\begin{align}
\label{eq:perrMono}
p_{err}^{(s)}(\cN(\rho),\cN(\sigma))\ge p_{err}^{(s)}(\rho,\sigma).
\end{align}
Moreover, $p_{\err}^{(s)}$ has a similar relationship to $Q^{(s)}$ as $p_{\err}$ has to $Q_{\min}$. In particular the following lemma shows that both hypothesis testing errors are equivalent, i.e. $p_{\err}^{(s)} \sim Q^{(s)}$. Moreover, we see that $p^{(s)}_{\err}$ is actually equivalent to $Q_C^{(s)}$ for all $C>0$ which establishes equivalence of all $Q_C^{(s)}$ for different values of $C>0.$

\begin{lemma}
\label{lem:QsCEquiv}
	For all $s,C>0$ we have
	\begin{align}
	Q^{(s)}_C \sim p_{\err}^{(s)}.
	\end{align}
	In particular this gives for all $C_1,C_2>0$
	\begin{align}
	Q^{(s)}_{C_1}\sim Q^{(s)}_{C_2}.
	\end{align}
	On the other hand, for $s=0$ and $C_1,C_2\in(0,1)$ with $C_1\neq C_2$ the hypothesis testing errors $Q^{(0)}_{C_1}$ and $Q^{(0)}_{C_2}$ are not equivalent. Equivalently this implies that for $\beta_{\eps_1} \nsim \beta_{\eps_2},$ for all $\eps_1,\eps_2\in(0,1)$ such that $\eps_1\neq\eps_2$.
\end{lemma}
\begin{proof}
	For any $C>0$, let us define 
	\begin{align}
	p^{(s,C)}_{\err}(\rho,\sigma) := \min_{0\le\Lambda\le\1} \Big(\left(\Tr\left((\1-\Lambda)\rho\right)\right)^{1/s} + C\Tr\left(\Lambda\sigma\right)\Big).
\end{align}
Clearly $p_{\err}^{(s)} \sim p_{\err}^{(s,C)}$, which follows by noting that
\begin{align*}
	p^{(s,C)}_{\err}(\rho,\sigma) \ge \min\{C,1\}\min_{0\le\Lambda\le\1} \Big(\left(\Tr\left((\1-\Lambda)\rho\right)\right)^{1/s} + \Tr\left(\Lambda\sigma\right)\Big) = \min\{C,1\}p_{\err}^{(s)}(\rho,\sigma)
\end{align*}
and 
\begin{align*}
p^{(s,C)}_{\err}(\rho,\sigma) \le \max\{C,1\} \min_{0\le\Lambda\le\1} \Big(\left(\Tr\left((\1-\Lambda)\rho\right)\right)^{1/s} + \Tr\left(\Lambda\sigma\right)\Big) = \max\{C,1\}p_{\err}^{(s)}(\rho,\sigma).
\end{align*}
\smallskip

Now we proceed to show $p_{\err}^{(s,C^{1/s})} \sim	Q^{(s)}_C$. The upper bound follows by
\begin{align}
p^{(s,C^{1/s})}_{\err}(\rho,\sigma) \nn&= \min_{0\le\Lambda\le\1} \Big(\left(\Tr\left((\1-\Lambda)\rho\right)\right)^{1/s} + C^{1/s}\Tr\left(\Lambda\sigma\right)\Big) \\&\le \min_{\substack{0\le\Lambda\le\1\\ \Tr\left((\1-\Lambda)\rho\right)\le C\left(\Tr(\Lambda\sigma)\right)^s}}\Big(\left(\Tr\left((\1-\Lambda)\rho\right)\right)^{1/s} + C^{1/s}\Tr\left(\Lambda\sigma\right)\Big) \nn\\&\le\min_{\substack{0\le\Lambda\le\1\\ \Tr\left((\1-\Lambda)\rho\right)\le C\left(\Tr(\Lambda\sigma)\right)^s}} 2\, C^{1/s}\Tr\left(\Lambda\sigma\right) = 2C^{1/s} Q^{(s)}_C(\rho\|\sigma).
\end{align}
The lower bound follows by first noting
\begin{align}
p^{(s,C^{1/s})}_{\err}(\rho,\sigma) \nn&= \min_{0\le\Lambda\le\1} \Big(\left(\Tr\left((\1-\Lambda)\rho\right)\right)^{1/s} + C^{1/s}\Tr\left(\Lambda\sigma\right)\Big)  \\&\nn\ge \min_{0\le\Lambda\le\1} \max\Big\{\left(\Tr\left((\1-\Lambda)\rho\right)\right)^{1/s}, C^{1/s}\Tr\left(\Lambda\sigma\right)\Big\} \\&= 
\min\Big\{ \min_{\substack{0\le\Lambda\le\1\\ \Tr\left((\1-\Lambda)\rho\right)\ge C\left(\Tr(\Lambda\sigma)\right)^s}}\left(\Tr\left((\1-\Lambda)\rho\right)\right)^{1/s}, \min_{\substack{0\le\Lambda\le\1\\ \Tr\left((\1-\Lambda)\rho\right)\le C\left(\Tr(\Lambda\sigma)\right)^s}}C^{1/s}\Tr\left(\Lambda\sigma\right)\Big\} \nn\\&= \min\Big\{\min_{\substack{0\le\Lambda\le\1\\ \Tr\left((\1-\Lambda)\rho\right)\ge C\left(\Tr(\Lambda\sigma)\right)^s}}\left(\Tr\left((\1-\Lambda)\rho\right)\right)^{1/s}, C^{1/s} Q_C^{(s)}(\rho\|\sigma) \Big\}.
\end{align}
Moreover,
\begin{align}
\nn\min_{\substack{0\le\Lambda\le\1\\ \Tr\left((\1-\Lambda)\rho\right)\ge C\left(\Tr(\Lambda\sigma)\right)^s}}\left(\Tr\left((\1-\Lambda)\rho\right)\right)^{1/s} &= \min_{\substack{0\le\Lambda\le\1\\ \left(\Tr\left((\1-\Lambda)\rho\right)/C\right)^{1/s}\ge\Tr(\Lambda\sigma)}}\left(\Tr\left((\1-\Lambda)\rho\right)\right)^{1/s} \\&= \left(Q^{(1/s)}_{1/C^{1/s}}(\sigma\|\rho)\right)^{1/s}.
\end{align}
Then using the fact that, by Lemma~\ref{lem:QsCequalities}, $ \left(Q^{(1/s)}_{1/C^{1/s}}(\sigma\|\rho)\right)^{1/s}=C^{1/s}Q_C^{(s)}(\rho\|\sigma),$ we get $$p^{(s,C^{1/s})}_{\err}(\rho,\sigma) \ge C^{1/s}Q_C^{(s)}(\rho\|\sigma).$$ Therefore, $Q_C^{(s)}(\rho\|\sigma)\sim p_{\err}^{(s,C^{1/s})}(\rho,\sigma)\sim p_{\err}^{(s)}(\rho,\sigma)$.
\smallskip

To finish the proof we show the non-equivalence of $\beta_{\eps_1}$ and $\beta_{\eps_2}$ for $\eps_1,\eps_2\in(0,1)$ with $\eps_1\neq \eps_2$. Here, without loss of generality $\eps_1>\eps_2.$ Take then for example $\sigma = \kb{0}$ and $\rho = \eps_1\kb{0} + (1-\eps_1) \kb{1}$. Then $\beta_{\eps_1}(\rho\|\sigma)= 0$ whereas $\beta_{\eps_2}(\rho\|\sigma)>0$. This follows by noting that an operator $0\le\Lambda\le\1$ satisfies $\Tr(\Lambda\sigma) = 0$ if and only if $\Lambda\le \kb{1}$ and hence $\Tr((\1-\Lambda)\rho) \ge \eps_1>\eps_2$.
\end{proof}
\medskip

We now have all the tools to prove that the optimal error exponent of $s$-hypothesis testing is equal to the divergence $\xi_s$ (Theorem~\ref{thm:AEPQs}). We first use the relation of the divergence $\xi_s$ with the quantum Hoeffding bound (cf. Lemma~\ref{lem:xisHoef}) to prove that the asymptotic error exponent of $p_{\err}^{(s)}$ is equal to $\xi_s$. Combining this with Lemma~\ref{lem:QsCEquiv} gives the desired result for the minimal error probability $Q_C^{(s)}$.

Let us recall the result of the quantum Hoeffding bound in a particular form which will be useful for the proof of Theorem~\ref{thm:AEPQs}: the optimal type I error exponent under the constraint that the type II error exponent is greater than or equal to $r$ is given by
\begin{align*}
B(r|\rho\|\sigma) &= \sup_{\substack{(\Lambda_n)_{n\in\N}\\0\le\Lambda_n\le\1}}\left\{\liminf_{n\to\infty}\frac{-\log\Big(\Tr\left((\1-\Lambda_n)\rho^{\otimes n}\right)\Big)}{n}\Bigg|\liminf_{n\to\infty}\frac{-\log\Big(\Tr\left(\Lambda_n\sigma^{\otimes n}\right)\Big)}{n} \ge r \right\}
\end{align*}
Hayashi \cite{Hayashi_Hoeffding(ErrorExponent)_2007} showed the achievability part of the Hoeffding bound which is that
\begin{align}
\label{eq:Hoef1}
B(r|\rho\|\sigma)\ge \sup_{0\le\alpha\le 1} \frac{\alpha-1}{\alpha}\left(r - D_\alpha(\rho\|\sigma)\right).
\end{align}
Nagaoka \cite{Nagaoka_ConverseHoeffding_2006} then proved the converse part. More precisely, from his work it follows\footnote{More precisely, Nagaoka stated the result \eqref{eq:Hoef2} but with limit inferior instead of limit superior applied on the sequence $\left(-\log(\Tr(\Lambda_n\sigma^{\otimes n}))/n\right)$ (see \cite[equations (6) and (11)]{Nagaoka_ConverseHoeffding_2006}). However, the result \eqref{eq:Hoef2} can be directly concluded from Nagaoka's work by modifying, in the obvious way, the inequality he used in his equation (25).} that for all $r>0$

\begin{align}
\label{eq:Hoef2}
&\nn\sup_{0\le\alpha\le 1} \frac{\alpha-1}{\alpha}\left(r - D_\alpha(\rho\|\sigma)\right) \\&\ge
\sup_{\substack{(\Lambda_n)_{n\in\N}\\0\le\Lambda_n\le\1}}\left\{\limsup_{n\to\infty}\frac{-\log\Big(\Tr\left((\1-\Lambda_n)\rho^{\otimes n}\right)\Big)}{n}\Bigg|\limsup_{n\to\infty}\frac{-\log\Big(\Tr\left(\Lambda_n\sigma^{\otimes n}\right)\Big)}{n} \ge r \right\}.
\end{align}
As the right-hand side of \eqref{eq:Hoef2} is lower bounded by $B(r|\rho\|\sigma)$, it follows that both inequalities \eqref{eq:Hoef1} and \eqref{eq:Hoef2} are actually equalities.

We will now prove our main result, given by Theorem~\ref{thm:AEPQs} in the Introduction. For convenience, we restate the result in more generality below.
\begin{theorem}
For all $s,C> 0$ and states $\rho,\sigma$ we have
	\begin{align}
	\lim_{n\to\infty}\frac{-\log\left(Q_C^{(s)}(\rho^{\otimes n}\|\sigma^{\otimes n})\right)}{n} = \lim_{n\to\infty} \frac{-\log\left(p^{(s)}_{\err}(\rho^{\otimes n},\sigma^{\otimes n})\right)}{n} = \xi_s(\rho\|\sigma).
	\end{align}
\end{theorem}
\begin{proof}
	We will show that
	\begin{align}
	\label{eq:perrsAEP}
	\lim_{n\to\infty} \frac{-\log\left(p^{(s)}_{\err}(\rho^{\otimes n},\sigma^{\otimes n})\right)}{n} = \xi_s(\rho\|\sigma),
	\end{align}
	since the result for $Q_C^{(s)}$ then directly follows by Lemma~\ref{lem:QsCEquiv}.
	Without loss of generality we assume that $\rho\neq \sigma$, since otherwise \comment{$Q_C^{(s)}(\rho^{\otimes n}\|\sigma^{\otimes n})= Q_C^{(s)}(\rho\|\sigma)$ and} $p^{(s)}_{\err}(\rho^{\otimes n},\sigma^{\otimes n})=p^{(s)}_{\err}(\rho,\sigma)$\footnote{This follows by the monotonicity of $p^{(s)}_{\err}$ under CPTP maps \eqref{eq:perrMono} applied for the maps $\cN(\cdot) =\Tr(\cdot) \rho$ and $\cM(\cdot) = \Tr(\cdot) \rho^{\otimes n}$.} and $\xi_s(\rho\|\sigma)=0$ and hence \eqref{eq:perrsAEP} is trivially true.
	For the achievability of \eqref{eq:perrsAEP}, using Lemma~\ref{lem:xisHoef} we can pick for any $\eps>0$ a sequence $\left(\Lambda_n\right)_{n\in\N}$ with $0\le\Lambda_n\le\1$ such that
	\begin{align*}
	\liminf_{n\to\infty}\frac{-\log\left(\Tr\left(\Lambda_n\sigma^{\otimes n}\right)\right)}{n} &\ge \xi_s(\rho\|\sigma),\\
	\liminf_{n\to\infty}\frac{-\log\left(\Tr\left((\1-\Lambda_n)\rho^{\otimes n}\right)\right)}{n} &\ge s\xi_s(\rho\|\sigma)-\eps.
	\end{align*}
	By~\eqref{25}
	\begin{align*}
	p^{(s)}_{\err}(\rho^{\otimes n},\sigma^{\otimes n}) &\le \left(\Tr\left((\1-\Lambda_n)\rho^{\otimes n}\right)\right)^{1/s}+ \Tr\left(\Lambda_n\sigma^{\otimes n}\right)\\  &\le 2 \max\Big\{\left(\Tr\left((\1-\Lambda_n)\rho^{\otimes n}\right)\right)^{1/s},\Tr\left(\Lambda_n\sigma^{\otimes n}\right)\Big\},
	\end{align*}
	and hence
	\begin{align*}
	\nn&\liminf_{n\to\infty} \frac{-\log\left(p^{(s)}_{\err}(\rho^{\otimes n},\sigma^{\otimes n})\right)}{n} \\&\ge\nn \min\left\{\liminf_{n\to\infty}\frac{-\log\left(\left(\Tr\left((\1-\Lambda_n)\rho^{\otimes n}\right)\right)^{1/s}\right)}{n},\liminf_{n\to\infty}\frac{-\log\left(\Tr\left(\Lambda_n\sigma^{\otimes n}\right)\right)}{n}\right\}\\&\ge \xi_s(\rho\|\sigma)-\eps/s.
	\end{align*}
	Since $\eps>0$ was arbitrary, this gives the achievability, i.e.
	\begin{align*}
	\liminf_{n\to\infty} \frac{-\log\left(p^{(s)}_{\err}(\rho^{\otimes n},\sigma^{\otimes n})\right)}{n} \ge \xi_s(\rho\|\sigma).
	\end{align*}
	For the converse assume 
	\begin{align*}
	\limsup_{n\to\infty} \frac{-\log\left(p^{(s)}_{\err}(\rho^{\otimes n},\sigma^{\otimes n})\right)}{n} >  \xi_s(\rho\|\sigma).
	\end{align*}
	Hence, there exists a sequence $\left(\Lambda_n\right)_{n\in\N}$ with $0\le\Lambda_n\le\1$ such that
	\begin{align*}
	\limsup_{n\to\infty}\frac{-\log\Big(\left(\Tr\left((\1-\Lambda_n)\rho^{\otimes n}\right)\right)^{1/s}+ \Tr\left(\Lambda_n\sigma^{\otimes n}\right)\Big)}{n} > \xi_s(\rho\|\sigma)
	\end{align*}
	and therefore in particular also
	\begin{align}
	\label{eq:rstar}
	r_\star \coloneqq 	\min\left\{\limsup_{n\to\infty}\frac{-\log\left(\left(\Tr\left((\1-\Lambda_n)\rho^{\otimes n}\right)\right)^{1/s}\right)}{n},\limsup_{n\to\infty}\frac{-\log\left(\Tr\left(\Lambda_n\sigma^{\otimes n}\right)\right)}{n}\right\}> \xi_s(\rho\|\sigma).
	\end{align}
	Hence, as $r\mapsto B(r|\rho\|\sigma)$ is monotonically decreasing, we see by Lemma~\ref{lem:xisHoef}
	
	\begin{align*}
	B(r_\star|\rho\|\sigma) \le B(\xi_s(\rho\|\sigma))|\rho\|\sigma) = s\xi_s(\rho\|\sigma) < sr_\star.
	\end{align*}
	However, from \eqref{eq:rstar} together with the discussion around \eqref{eq:Hoef2} we find
	
	\begin{align*}
	B(r_\star|\rho\|\sigma)\ge \limsup_{n\to\infty}\frac{-\log\Big(\Tr\left((\1-\Lambda_n)\rho^{\otimes n}\right)\Big)}{n} =  s\limsup_{n\to\infty}\frac{-\log\left(\left(\Tr\left((\1-\Lambda_n)\rho^{\otimes n}\right)\right)^{1/s}\right)}{n}  \ge s r_\star
	\end{align*}
	which is a contradiction. Hence,
	\begin{align*}
	\limsup_{n\to\infty} \frac{-\log\left(p^{(s)}_{\err}(\rho^{\otimes n},\sigma^{\otimes n})\right)}{n} \le \xi_s(\rho\|\sigma).
	\end{align*}
\end{proof}

\medskip

\noindent\textbf{Acknowledgments.} ND would like to thank Yury Polyanskiy for helpful comments.
RS gratefully acknowledges support from the Cambridge Commonwealth, European and International Trust.
\appendix
\section{Finiteness of the quantum Hoeffding bound}
We prove the necessary and sufficient condition for finiteness of $B(r|\rho\|\sigma)$ mentioned in the Introduction.
\begin{lemma}
\label{lem:AppFinHoeff}
Let $\rho,\sigma$ states and $r\ge0$. Then $B(r|\rho\|\sigma)<\infty$ if and only if $r>D_{\min}(\rho\|\sigma).$ Here, $B(r|\rho\|\sigma)$ is defined in \eqref{eq:HoeffdingDef} and $D_{\min}(\rho\|\sigma)= -\log(\Tr(\pi_\rho \sigma))$ with $\pi_\rho$ being the projector onto the support of $\rho$.
\end{lemma}
\begin{proof}
Let $r>D_{\min}(\rho\|\sigma)$. As $\lim_{\alpha\to 0 }D_\alpha(\rho\|\sigma) = D_{\min}(\rho\|\sigma)$ there exists $\alpha_0>0$ such that $r-D_\alpha(\rho\|\sigma)\ge 0$ for all $0\le\alpha\le \alpha_0$. As also $r>0$ by assumption, the relation \eqref{eq:HoeffResult} can be employed which gives 
\begin{align}
\label{eq:FiniteHoeffCalc}
\nn B(r|\rho\|\sigma) &= \sup_{0\le\alpha\le 1}\frac{\alpha-1}{\alpha}\left(r-D_{\alpha}(\rho\|\sigma)\right) =  \sup_{\alpha_0 \le\alpha\le 1}\frac{\alpha-1}{\alpha}\left(r-D_{\alpha}(\rho\|\sigma)\right) \\&\le \frac{1}{\alpha_0}\sup_{\alpha_0 \le\alpha\le 1}\left|(\alpha-1)\left(r-D_{\alpha}(\rho\|\sigma)\right)\right| =  \frac{1}{\alpha_0}\sup_{\alpha_0 \le\alpha\le 1}\left|(\alpha -1)r-\log(\Tr(\rho^\alpha\sigma^{1-\alpha}))\right|.
\end{align}
Since $r>D_{\min}(\rho\|\sigma)$ we see in particular $D_{\min}(\rho\|\sigma)<\infty$ and hence $\rho$ and $\sigma$ have mutually non-orthogonal supports. By that $\sup_{\alpha_0\le\alpha\le 1}|\log(\Tr(\rho^\alpha\sigma^{1-\alpha}))|<\infty$ which together with \eqref{eq:FiniteHoeffCalc} gives $B(r|\rho\|\sigma)<\infty$.

As $r\mapsto B(r|\rho\|\sigma)$ is monotonically decreasing, for the other direction it suffices to show that for $r=D_{\min}(\rho\|\sigma)$ we have $B(r|\rho\|\sigma) =\infty$.
Let for that $n\in\N$ and $\Lambda_n = \pi_{\rho^{\otimes n}}$. By definition, $$\frac{-\log(\Tr(\Lambda_n\sigma^{\otimes n}))}{n} = \frac{D_{\min}(\rho^{\otimes n}\|\sigma^{\otimes n})}{n}  = D_{\min}(\rho\|\sigma)$$ and $-\log(\Tr((\1-\Lambda_n)\rho^{\otimes n}))/n = \infty$ for all $n\in\N$, which gives $B(D_{\min}(\rho\|\sigma)|\rho\|\sigma) =\infty$ and finishes the proof.
\end{proof}

\nocite{PolyanskiyWu_DissipationOfInformation_2016}
\bibliography{Ref}
\bibliographystyle{abbrv}
\end{document}